\documentclass[11pt,a4paper]{article}
\usepackage[margin=1in]{geometry}

\usepackage{setspace}
\usepackage{amsthm}
\usepackage{amsmath}
\usepackage{amsfonts}
\usepackage{graphicx}
\usepackage{hyperref}
\usepackage{color}
\usepackage{caption}
\usepackage{subcaption}
\usepackage{float}
\usepackage{soul}
\usepackage{multirow}
\usepackage{xspace}
\usepackage{todonotes}
\usepackage[inline]{enumitem}
\usepackage{microtype,amssymb}
\usepackage{cleveref}
\usepackage{thm-restate}
\usepackage{framed}
\usepackage[ruled]{algorithm}
\usepackage[noend]{algpseudocode}

\algrenewcommand\textproc{}
\usetikzlibrary{matrix}
\newtheorem{theorem}{Theorem}
\newtheorem{lemma}[theorem]{Lemma}

\newtheorem{corollary}[theorem]{Corollary}
\newtheorem{definition}[theorem]{Definition}
\newtheorem{prp}[theorem]{Proposition}

\newtheorem{remark}[theorem]{Remark}

\newtheorem{problem}{Problem}
\newtheorem{example}{Example}
\newtheorem{observation}[theorem]{Observation}

\def\RR{{\mathbb R}}
\def\NN{{\mathbb N}}

\def \eps{\varepsilon}

\def \rank{\mathrm{rank}}

\def \vectorset{\mathrm{set}}

\DeclareMathOperator{\feas}{D}
\DeclareMathOperator{\hmin}{a}
\DeclareMathOperator{\hmax}{b}

\title{A near-linear time approximation scheme for $(k,\ell)$-median clustering under discrete Fr\'echet distance}

\author{
Anne Driemel\thanks{University of Bonn. Affiliated with Lamarr Institute for Machine Learning and Artificial Intelligence.}
\and
Jan Höckendorff\thanks{Department of Mathematics and Computer Science, University of Cologne. Funded by the Deutsche Forschungsgemeinschaft (DFG, German Research Foundation) – Project Number 459420781.}
\and Ioannis Psarros \thanks{Archimedes, Athena Research Center. Partially supported by project MIS 5154714 of the National Recovery and Resilience Plan Greece 2.0 funded
by the European Union under the NextGenerationEU Program.}
\and Christian Sohler \thanks{Department of Mathematics and Computer Science, University of Cologne.}
}

\date{}

\begin{document}

\maketitle
\begin{abstract}
A time series of complexity $m$ is a sequence of $m$ real valued measurements.
The discrete Fréchet distance $d_{dF}(x,y)$ is a distance measure between two time series $x$ and $y$ of possibly different complexity. 
Given a set of $n$ time series represented as $m$-dimensional vectors over the reals, the $(k,\ell)$-median problem under discrete Fréchet distance
aims to find a set $C$ of $k$ time series of complexity $\ell$ such that 
$$\sum_{x\in P} \min_{c\in C} d_{dF}(x,c)$$ is minimized. In this paper, we give the first near-linear time $(1+\eps)$-approximation algorithm for this problem when $\ell$ and $\eps$ are constants but $k$ can be as large as $\Omega(n)$. 
We obtain our result by introducing a new dimension reduction technique for discrete Fréchet distance and then adapt an algorithm of Cohen-Addad et al. \cite{CFS21} to work on the dimension-reduced input. As a byproduct
we also improve the best coreset construction for $(k,\ell)$-median under discrete Fréchet distance
\cite{CDRSS25} and show that its size can be independent of the number of input time series \emph{ and } their complexity.
\end{abstract}
\section{Introduction}

Clustering is the process of grouping a set of elements into subsets called clusters in such a way that, ideally, elements within a subset are similar to each other and elements that are not in the same cluster are not similar. Clustering is a method from unsupervised learning with a wide range of applications
\cite{CFS21}.
In order to be able to perform  clustering, one needs a distance or similarity measure that allows to compare different elements. The quality of
a clustering depends heavily on how well a similarity measure reflects the
ground truth similarity of the objects. Once we have decided on the particular distance or similarity measure to be used, a common approach is to
formulate the clustering problem as finding a partition that minimizes some cluster objective function. Often, clusters are represented by a cluster center and the quality of a cluster is determined by the distances of its members to the center. 

In the theoretical computer science community some of the most frequently studied center based clustering methods are $k$-median, $k$-means and $k$-center clustering. In the $k$-median problem we aim to find a set of $k$ centers, such that the sum of distances to the nearest center is minimized. In the $k$-means variant, the sum of squared distances is to be minimized and in the $k$-center version, we aim at minimizing the maximum distance to the nearest center.
In all variants, the desired number of clusters $k$ is given as input. 
Often, it is assumed that the underlying distance measure is a metric. For general metric spaces, one usually considers the discrete $k$-median problem where the centers are supposed to be part of a set of potential candidates $F$ and the set of points to be clustered is $C$. The currently best algorithm is a randomized polynomial time approximation algorithm that guarantees a $(2+\eps)$-approximation \cite{CLSS25}.
For metrics of bounded doubling dimension
there is a linear time approximation scheme that with constant probability computes a $(1+\eps)$-approximation \cite{CFS21}.
For specific metric spaces such as the Euclidean metric one also allows the centers to be arbitrary points of the metric space. Here one usually aims at computing a $(1+\eps)$-approximation as well \cite{CFS21}.

 In this paper, we will consider a variant of $k$-median clustering under the discrete Fréchet distance. The discrete Fréchet distance originates from a distance measure  defined by Maurice Fréchet in the context of parametrized curves~\cite{frechet1906quelques,deza2009encyclopedia}.
 
 A standard way to introduce the  Fréchet distance uses the picture of a person walking a dog, where the dog moves along one curve and the person moves along the other curve, each at their own varying speed. At each point in time, the person and the dog are connected by a leash. The Fréchet distance is the minimum length of a leash that permits a walk from the beginning of the curves to the end without backtracking. 
 In this paper, we will focus on the 
 discrete Fréchet distance. The pointwise distances contributing to the discrete Fréchet distance are measured only between the discrete points of the input time series. In keeping with above metaphor, the dog and its owner are required to jump from one point to the next along the time series.
 Furthermore, we will only consider the
 case $d=1$, i.e. each time series is a sequence of real values.

The discrete Fréchet distance (as well as its continuous counterpart) is motivated 
by problems in data analysis when time series come from measurements that are not well aligned. That is, 
assume that a time series results from measurements of a continuous signal at certain points of time, then the discrete Fréchet distance tolerates different and/or non-regular sampling rates better than classical distance measure like squared Euclidean distance and is therefore a useful alternative. 
This motivates the problem of clustering time series under Fréchet distance.
Unfortunately, the time series that minimizes the sum of distances to a given set of $n$ time series of length $m$ may have a complexity as large as $\Omega(mn)$, which leads to unwanted effects such as overfitting. For this reason, the $(k,\ell)$-median problem 
has been introduced \cite{DKS16} for the related continuous Fréchet distance and was later also considered for the discrete Fréchet 
distance \cite{nath2021k}. It is known that the Fréchet distance is a (pseudo-)metric and so one can apply algorithms for general metric spaces and get a constant approximation using suitable simplifications of the input curves as a candidate set for the centers.
The doubling dimension of the discrete Fréchet distance of real-valued time series of length at most $m$ can be $\Omega(m)$, so results for metrics of bounded doubling dimension do not give a polynomial time algorithm. 
For the $(k,\ell)$-median problem, there is a $(1+\eps)$-approximation algorithm that is near linear  in $n$ and $m$ but exponential in $k$ \cite{nath2021k} (but it also works for ambient dimension $d>1$).

The main result of this paper is a near-linear time approximation scheme for $(k,\ell)$-median clustering for $n$ real-valued time series of length $m$ under discrete Fréchet distance, where $\ell$ and $\eps$ are assumed to be constants and polynomial running time in $k,m$ and $n$.

\subsection{Related work}

There is extensive literature on the complexity of computing the Fréchet distance. Computing the distance for two curves of complexity $m$ takes roughly $O(m^2)$ time both for the discrete and the continuous version~\cite{alt1995computing, agarwal2014computing}. This is optimal under the strong exponential time hypothesis~\cite{B14,abboud2018tighter}, even for curves in 1-dimensional ambient space ~\cite{bringmann2016approximability,buchin2019seth}, and it holds even if we allow a (small) constant approximation.

The $k$-median problem in metric spaces is known to be hard to approximate with a factor better than $1+2/e$ unless set cover can be approximated within a factor $c \ln n$ for $c<1$ \cite{JMS02}.
 A number of different constant factor polynomial time approximation algorithms are known \cite{cowen2003constant,jain2001approximation,arya2004local,gupta2008simpler,cohen2022improved,mettu2003online,charikar2012dependent,thorup2005}, and the currently best approximation ratio is $(2+\eps)$ \cite{CLSS25}.

In the Euclidean plane it is NP-hard \cite{megiddo1984complexity}.
The first polynomial time approximation scheme for $k$-median in the Euclidean plane has been developed by Arora et al. \cite{arora1998approximation} and later improved to near-linear time in $\mathbb R^d$ when $d$ is constant \cite{kolliopoulos2007nearly} . This result has been generalized to metric spaces of bounded doubling dimension \cite{T04}
         and later to a near-linear approximation scheme  \cite{CFS21}.
   
 Driemel et al. \cite{DKS16} defined the $(k,\ell)$-clustering problem for time series as follows: Given a set $P$ of $n$ time series of complexity $m$ and parameters $k,\ell \in \NN$ find $k$ center time series of complexity $\ell$, such that (a) the maximum distance of an element in $P$ to its closest center time series or (b) the sum of these distances is minimized. Variant (a) is referred to  as $(k,\ell)$-center and (b) as $(k,\ell)$-median. Under the continuous Fréchet distance, they developed near-linear time $(1+\eps)$-approximation algorithms for both clustering variants, assuming $\eps,k$, and $\ell$ are constants. They complement these algorithmic results with hardness results, showing that both $(k,\ell)$-median and $(k,\ell)$-center are NP-hard under continuous Fréchet distance. Approximating $(k,\ell)$-median for polygonal curves in arbitrary dimensions was recently studied in \cite{BDR23}. Cheng and Huang give the first $(1+\eps)$-approximation algorithm for $(k,\ell)$-median under continuous Fréchet distance in $d > 1$~\cite{CH23}. Both clustering problems are also NP-hard under the discrete Fréchet distance and even for the case $k = 1$ \cite{BDGHKLS19} \cite{buchin2019hardness}. 
 Buchin et al. developed the first $(1+\eps)$-approximation algorithm for $(k,\ell)$-median under discrete Fréchet distance, which runs in polynomial time assuming $k$ to be constant. Nath and Taylor \cite{nath2021k} improved this to near-linear time for constant $k$. 
Buchin and Rohde \cite{buchin2022coresets} designed the first coreset construction for $(k,\ell)$-median under both variants of the Fréchet distance, where the size of the coreset has logarithmic dependence on the number of input curves.
Recently, Cohen-Addad et al. introduced a coreset construction for $(k,\ell)$-median under discrete Fr\'echet distance that has size independent of the number of input curves~\cite{CDRSS25}.

Related to our dimension reduction are some data structures for approximate nearest neighbor search under discrete Fr\'echet distance \cite{driemel2019sublinear, filtser2023static,FFK23}. In the asymetric setting where the query time series has complexity $\ell$, the data structures cited above replace each input time series by a set of lower dimensional time series. This is fundamentally different from our dimension reduction, which replaces each time series with exactly one lower dimensional time series.

\section{Technical Overview}

The main result of this paper is a near-linear time $(1+\eps)$-approximation algorithm for $(k,\ell)$-median clustering under the discrete Fréchet distance, where the parameters $\ell$
and $\eps$ are considered to be constants. All prior algorithms have a running time exponential in the number of clusters $k$ \cite{BDGHKLS19, nath2021k}.

\begin{theorem} (Informal version of \Cref{theorem:final_formal})
Let $\eps \in (0,1/2]$ and $\ell \in \NN$
be constants. Given a set of $n$ real-valued time series of complexity $m$, and parameters $\eps, \ell$ and $k$,
Algorithm \ref{alg:ltas} computes in time
$\tilde O(mn)$ and with success probability at least $1-\eps$ a $(1+\eps)$-approximation to the $(k,\ell)$-median problem under discrete Fréchet distance. 
\end{theorem}

The $\tilde O$-notation drops here an $\log (nm)^{O(1)}$ factor. 

\subsection{High level approach} 
Our high-level approach to the problem is as follows. We first apply a new near-linear time dimension reduction algorithm that maps the input time series of complexity $m$ to a low-dimensional space (that is, to time series of low complexity) whose dimension depends only on the parameters $\ell$ and $\eps$ and that preserves distances to time series of complexity at most $\ell$ up to a factor of $(1\pm \eps)$.
Since the doubling dimension of the discrete Fr\'echet distance of curves of complexity at most $m$ is $O(m)$ and because the target dimension of our reduction depends only on the constant parameters $\eps$ and $\ell$, we have reduced our problem to a $(k,\ell)$-clustering problem with bounded doubling dimension. We then show how to adapt a linear time approximation scheme by Cohen-Addad et al. \cite{CFS21} 
for the $k$-median problem with bounded doubling dimension to our problem. Since their algorithm already allows to specify a set of facilities from which the centers are to be chosen, it suffices to construct a sufficiently small set of candidate points that contains a $(1+\eps)$-approximation and apply their algorithm on the dimension-reduced point set and set of candidate points.

\subsection{Dimension reduction}
Our main technical contribution is a new dimension reduction for the discrete Fr\'echet distance that can be summarized as follows.

\begin{restatable}{theorem}{thmdimreduction}{(Informal version of \Cref{theorem:dimreduction})}

Let $\varepsilon \in (0,1)$ and $\ell \in \NN$ be constants. There exists a constant $d_0 = d_0(\varepsilon, \ell)$ such that for arbitrary $m\in \NN$ and input time series $x \in \RR^m$ Algorithm \ref{alg: ComplexityReduction} computes in time $O(m \log^2 m)$ 
a time series $z \in \RR^{d_0}$ s.t. for all $y \in \RR^\ell$,
\begin{equation*}
    (1-\varepsilon) d_{dF}(x,y) \leq d_{dF}(z,y) \leq (1 + \varepsilon)  d_{dF}(x,y).
\end{equation*}
\end{restatable}

Our idea for the dimension reduction can be described as follows.
As a first (simple) step, we show that one can reduce the number of distinct values appearing in a time series of complexity $m$ to $O(\ell/\eps)$ while maintaining the distance to any time series of complexity $\ell$ up to a factor of $(1+\eps)$. 

Then we observe that the discrete Fr\'echet distance between two time series $x$ and $y$, each of fixed complexity, can be written as a minimum over all traversals. Our goal is to describe the function minimizing over all traversals with a function minimizing over a much smaller set.
During each traversal, every value $y_i$ of $y$ is matched to a subsequence of $x$. In order to determine the Fr\'echet distance, it suffices to know the minimum and maximum value of $x$ matched to $y_i$. 
As it turns out, each traversal can  equivalently (but not uniquely!) be described by remembering a sequence of constraints that consist of the minimum and maximum value matched to each $y_i$. Furthermore, the function minimizing over the set of all possible traversals can likewise be described by minimizing over all possible ordered constraint sets. 
Since we have reduced the number of different values of the time series to $O(\ell/\eps)$ the number of different constraint sets is small. However, it still depends on the set of different values of the time series $x$.

To remove this dependence we remember for each traversal and each $y_i$ the \emph{rank} of the minimum and maximum value matched to it with respect to the (reduced) set of values of $x$. The resulting set of constraints 
will be called an $\ell$-profile (see Figure \ref{fig: example profile} for an example).
It is important to note that the set of all $\ell$-profiles
of a time series $x$ with values from a fixed set $X$
completely determines the Fr\'echet distance
to any time series of complexity $\ell$.
Since there are only a constant number  of different sets of $\ell$-profiles (where the constant depends on $\ell$ and $\eps$) for any time series $x$ with $O(\ell/\eps)$ distinct values, we can replace $x$ by the shortest time series $z$ over the same set of values that has the same set of $\ell$-profiles. 

What is the complexity of $z$?
We first observe that we can also replace any other time series $x'$ with the same set of $\ell$-profiles and the same number of distinct values with $z$, if we replace its values by the corresponding values from the set of values of $x'$.
Now, we can group all time series (of arbitrary length) according to their number of distinct values and set of profiles. For each group, the length of the shortest such time series is a constant that depends on the set of profiles and the number of distinct values of the time series. 
Since the number of profiles is also a constant depending on $\eps$ and $\ell$, the maximum length of these shortest time series is constant as well. Since the number of distance values also only depends on $\ell$ and
$\eps$, we know we can write this constant as a function of $\ell$ and $\eps$. We remark that we do not know any closed form or upper bound on this function.

\begin{remark}
In pursuit of a constructive bound, it is tempting to try a simple greedy pruning procedure that removes values from $x$ while maintaining the same set of $\ell$-profiles. While it is easy to make sure that no $\ell$-profile vanishes, we do not know how to avoid that new $\ell$-profiles are created using such a procedure.
\end{remark}

\subsection{How to apply the approximation scheme of  Cohen-Addad et al. \cite{CFS21} }

After the dimension reduction, we obtain a space that has reduced (and bounded) doubling dimension and we would like to apply 
a linear time approximation scheme by Cohen-Addad et al. \cite{CFS21} on the dimension-reduced time series. Their algorithm allows to specify a distinct set of clients (points to be clustered) and facilities (points where one may put a center). This is good for us, since we need to ensure a bound on the complexity of centers chosen by the algorithm. However, their algorithm is linear in the size of the union of both sets. In our case, the set of time series of complexity $\ell$ is unbounded. We therefore need to construct a \emph{center set} $F$ that is sufficiently small and that contains a $(1+\eps)$-approximation.
To do so we utilize a technique by Filtser et al. \cite{FFK23} that was developed for the context of approximate nearest neighbor data structures.

\subsection{Additional Results}
Our dimension reduction can also be used to improve coreset constructions for the $(k,\ell)$-median problem. For example we can combine the dimension reduction with the result of Cohen-Addad et al.~\cite{CDRSS25} to obtain the first coreset for the $(k,\ell)$-median problem of time series under the discrete Fr\'echet distance whose size is completely independent of the input.

\section{Preliminaries}
\label{section:preliminaries}
We will represent a time series of length $m$ as an $m$-dimensional vector $x \in \RR^m$ over the reals.
We also refer to $m$ as the \emph{complexity} of the time series. For a positive integer $r$ we will use $[r]$ to denote the set $\{1,2,\dots,r\}$.
To define the discrete Fréchet distance  between two time series $x,y$ we first introduce the notion of traversals, which can be seen as an alignment between the vertices of $x$ and $y$ that satisfies certain properties. 
 
\begin{definition}
Given two time series $x = (x_1,\dots,x_{m}) \in \RR^{m}$ and $y = (y_1,\dots,y_{\ell}) \in \RR^{\ell}$ a \emph{traversal} $T$  is a sequence of index pairs $(i,j) \in [m]\times[\ell]$ with the following properties:
\begin{enumerate}
\item $T$ starts with pair $(1,1)$
\item $T$ ends with pair $(m,\ell)$ 
\item if a pair $(i,j)$  appears in $T$, it can only be followed by either $(i,j+1)$, $(i+1,j)$ or $(i+1,j+1)$
\end{enumerate}
If a pair $(i,j)$ appears in $T$, we say that $x_i$ and $y_j$ are \emph{matched} in traversal $T$.
\end{definition}

Next, we will introduce a matrix based notation for traversals that we find useful for presenting our results.
 In order to do so, observe that traversals effectively extend both time series individually to some common length by repeating arbitrary vertices. This extension mechanism can also be realized by a matrix multiplication, i.e. given a time series $x$ of complexity $\ell$ and a suitable matrix $M \in \{0,1\}^{t\times \ell}$ we obtain a new time series $M \cdot x$ of complexity $t$. We will call these matrices \emph{traversal matrices}.
Thus, a traversal matrix is a matrix that maps an $\ell$-dimensional vector $v$ in a $t$-dimensional space by copying some of the entries of $v$ while keeping their relative order.
For a matrix $M$ we will use $M(i,j)$ to denote the entry in
row $i$ and column $J$. We define traversal matrices as follows.

\begin{definition}[traversal matrices]
$M \in \{0,1\}^{t \times \ell}$ is a traversal matrix, if it satisfies the following properties:

\begin{enumerate}
    \item The rows of $M$ are standard unit vectors,
    \item $M(1,1) = 1$,
    \item $M(t,\ell) = 1$, and
    \item For rows $i \in [t-1]$ and columns $ j \in [\ell-1]$ it holds that if $M(i,j) = 1$ then either $M(i+1,j+1) = 1$ or $M(i+1,j) = 1$. Furthermore, if $i\in [t-1]$ and $j=\ell$ then $M(i+1,j)=1$.
\end{enumerate}
\end{definition}

For $t,\ell \in \NN$ let $\mathcal{M}^t_\ell
\subset \{0,1\}^{t\times m}$ be the set of all $(t,\ell)$-dimensional traversal matrices.
When the dimensions $t$ and $\ell$ are clear from the context, we will not mention $\mathcal{M}^t_{\ell}$ explicitly. Observe that with the above definition, we always have $t \ge \ell$.
We say that a pair of traversal matrices 
$(M_m, M_\ell)$
has \emph{matching dimension}
for two time series $x\in \mathbb R^m$ and $y\in \mathbb R^\ell$, if
there is $t \in \mathbb N$ such that $t\ge\max(m,\ell)$,
$M_m \in \mathcal M_{m}^t$, and $M_\ell\in \mathcal M_{\ell}^t$. If
$x$ and $y$ are clear from the context, we also just say that $(M_m,M_\ell)$ has matching dimension.
We observe that any traversal of $x=(x_1,\dots, x_m)$ and $y=(y_1,\dots, y_\ell)$ can be represented as a pair of traversal matrices
$(M_m, M_\ell)$ of matching dimension. 
Furthermore, 
note that the $t$-dimensional vectors 
$M_m \cdot (1,2,\dots, m)^T$ and $M_\ell \cdot (1,2,\dots, \ell)^T$
can be used to denote the indices of the entries
of $x$ and $y$ that are mapped to the corresponding 
entries in the $t$-dimensional vectors $M_m x$ and $M_\ell y$, respectively. This way 
we can identify $(M_m, M_\ell)$ with a sequence $T'$ of index pairs whose $r$-th entry
is the pair $(i,j)$ where $i$
is the $r$-th entry of 
$M_m \cdot (1,2,\dots, m)^T$ and $j$ is the $r$-th entry of $M_\ell \cdot (1,2,\dots, \ell)$.
If we remove multiple neighboring occurences of pairs $(i,j)$ from $T'$ the resulting sequence $T$
is a traversal. Hence, 
we can identify any pair
$(M_m,M_\ell)$ of matrices of matching dimension with a traversal.
We can therefore also say
that $x_i$ and $y_j$
are matched by $(M_m,M_\ell)$, if $(i,j)$
appears in the corresponding traversal.
We can then define the Fr\'echet distance in the following way. 

\begin{definition}
Given time series $x \in \RR^m,y \in \RR^{\ell}$ the 
 \emph{discrete Fréchet distance} 
 $d_{dF}(x,y)$ between $x$ and $y$ is defined as
    \[d_{dF}(x,y) = \min_{(M_m,M_\ell)} \left \Vert M_m x - M_\ell  y \right \Vert_\infty,\]
where the minimum is over all pairs of traversal matrices $(M_m,M_\ell)$
of matching dimension.
\end{definition}

\begin{figure}[t!]
    \begin{subfigure}[t]{0.48\textwidth}
    \centering
       \includegraphics[ page = 1]{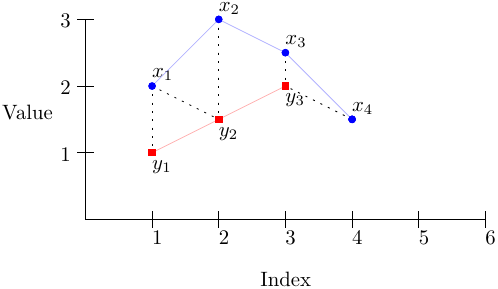}
        \caption{
        }
         \label{fig:1.a}
    \end{subfigure} 
    \begin{subfigure}[t]{0.48\textwidth}
    \centering
        \includegraphics[page = 2]{figures/traversal_matrix_example.pdf}
        \caption{
        }
         \label{fig:1.b}
    \end{subfigure}
    \caption{Subfigure \ref{fig:1.a} shows two time series $x \in \RR^4$ and $y \in \RR^3$ matched by some traversal with corresponding traversal matrices $(M_4,M_3)$. The dotted lines indicate which vertices are matched by the traversal. Subfigure \ref{fig:1.b} shows $x':= M_4 \cdot x$ and $y':= M_3 \cdot y$.}
\end{figure}

\begin{example}
Consider two time series $x\in \RR^4$, $y\in \RR^3$ and a traversal with corresponding traversal matrices \[M_4 = \begin{pmatrix}
    1 & 0 & 0 & 0\\
    1 & 0 & 0 & 0\\
    0 & 1 & 0 & 0\\
    0 & 0 & 1 & 0\\
    0 & 0 & 0 & 1\\
\end{pmatrix},
M_3 = \begin{pmatrix}
    1 & 0 & 0 \\
    0 & 1 & 0 \\
    0 & 1 & 0 \\
    0 & 0 & 1 \\
    0 & 0 & 1 \\
\end{pmatrix}\]  as depicted in Figure \ref{fig:1.a}.  Figure \ref{fig:1.b} shows the two time series $x$ and $y$ after applying the traversal matrices, i.e. $x':= M_4\cdot x$ and $y':= M_3 \cdot y$.
\end{example}

\subsection{The $(k,\ell)$-median clustering problem}

In this paper, we consider the $(k,\ell)$-median problem for clustering under the discrete Fr\'echet distance that has been introduced in \cite{DKS16} in the context of the continuous Fr\'echet distance. 
The problem is a variant of the $k$-median problem under Fr\'echet distance, where the complexity of the center time series is restricted to be at most $\ell$. This restriction is motivated by the fact that otherwise an optimal center time series may be as large as $\Omega(nm)$, where $n$ is the number of input time series and $m$ is their complexity.
\begin{problem}[($k,\ell)$-median clustering]
   Given a set of time series $P \subset \RR^m$ and parameters $k, \ell \in \NN$, compute a set $C$ of $k$ time series  
   of complexity at most $\ell$
   that minimizes
    \[\sum_{x \in P}\min_{c \in C} d_{dF}(x,c).\] 
\end{problem}
We also remark that one may also define the center time series to have complexity exactly $\ell$ without changing the problem, since every time series with complexity fewer than $\ell$ can be extended to a time series with $\ell$ vertices by repeating the first element of the time series. This does not affect the Fr\'echet distance.

\subsection{Minimum error simplification}

In our paper we use the concept of minimum-error $\ell$-simplification for some time series $x \in \RR^m$, which can be computed in $O(\ell  m  \log(m)  \log(m/\ell))$ time as in \cite{bereg2008simplifying} and is defined as follows.

\begin{definition}[minimum-error $\ell$-simplification]
For a time series $x \in \RR^m$ a time series $\tilde{x} \in \RR^{\ell}$ is a minimum-error $\ell$-simplification of $x$ if for any time series $y \in \RR^{\ell}$ it holds that $d_{dF}(x,\tilde{x}) \leq d_{dF}(x,y)$.
\end{definition}

\section{Dimension reduction}\label{sec:dimreduction}
In this section, we will establish our dimension reduction result. In a first step, we will show that one can quantize the entries of a time series from $\mathbb R^m$
to $O(\ell/\varepsilon)$ distinct values, while  guaranteeing that the distance to any time series of complexity $\ell$ is preserved up to a factor of $(1\pm \varepsilon)$.

\begin{algorithm}
\caption{\label{alg: ReduceValueDomain}}
\begin{algorithmic}[1]

\Procedure {\texttt{ReduceValueDomain}}{$x \in \RR^m$, $\ell \in \NN$, $\varepsilon \in \RR_{> 0}$}

\State Let $\tilde{x}$ be a minimum-error $\ell$-simplification of $x$

\State $\Delta \gets d_{dF}(x,\tilde{x})$

\State Obtain $x'$
from $x$ by rounding up the entries of $x$ to the next multiple of $\varepsilon\Delta$ 

\State \Return $x'$
\EndProcedure
\end{algorithmic}
\end{algorithm}

\begin{restatable}{lemma}{lemreducevaluedomain}\label{lem: reduce value domain}
Let $x \in \RR^m$, $\ell \in \NN$ and $1\ge\varepsilon>0$. Algorithm \ref{alg: ReduceValueDomain} computes in time $O(\ell m  \log(m)  \log(m/\ell))$ a time series $x' \in X^m$ with $X \subset \RR$ and $|X| \in O(\ell /\varepsilon)$ s.t. for every $y \in \RR^\ell$, \[ (1-\varepsilon) d_{dF}(x,y) \le d_{dF}(x',y) \le (1+ \varepsilon)  d_{dF}(x,y).\]
\end{restatable}

\begin{proof}
   The minimum-error $\ell$-simplification $\tilde x$ of $x$ can be computed in $O(\ell m \log (m) \log (m/\ell))$ time using an algorithm from 
\cite{bereg2008simplifying}
and the distance $d_{dF}(x,\tilde x)$ can be computed in $O(\ell m)$ time.
 The remaining steps are linear $m$ and hence the running time follows.
Now let $\tilde{x}$ be the computed minimum-error $\ell$-simplification of $x$ and let $\Delta = d_{dF}(x,\tilde{x})$
. 
    For $1 \leq j \leq \ell$ we define $S_j = [ \tilde{x}_j - \Delta, \tilde{x}_j + (1+\varepsilon)\Delta ]$. Then it holds for all $1 \leq i \leq m$ that there is $1 \leq j \leq \ell$ s.t. $x_i \in S_j$. It follows that the number of distinct values of $x'$
    is $O(\ell/\varepsilon)$. 
By the triangle inequality we have 
$
d_{dF}(x',y) \le d_{dF}(x',x) + d_{dF}(x,y)
$ and 
$d_{dF}(x',y) \ge d_{dF}(x,y) - d_{dF}(x',x)$.
By the trivial alignment we have
$d_{dF}(x',x) \le \epsilon \Delta$. The lemma follows from the fact that by definition of a minimum-error $\ell$-simplification we have $\Delta = d_{dF} (\tilde{x},x) \leq d_{dF}(y,x)$
    for every $y \in \RR^\ell$.
\end{proof}

In our dimension reduction we are interested in maintaining 
the discrete Fr\'echet distance of a time series $x\in \RR^m$
to every time series of length at most $\ell$. Thus, we can use the previous lemma to reduce the number of distinct values of every fixed time series to $O(\ell/\varepsilon)$. We will 
therefore focus in the remainder of this chapter on such time series and exploit this property for our dimension reduction.

For $x\in \RR^m$, $x= (x_1,\dots, x_m)$, define
$\vectorset(x) = \{x_i : 1\le i \le m\}$.
For $x \in \RR^m$ and $z\in \vectorset(x)$ let $
\rank_x(z)$ be the rank of $z$ in $\vectorset(x)$, the set  of entries of $x$.
Note that one can compute the rank of all elements in $\vectorset(x)$ in $O(|\vectorset(x)|\log |\vectorset(x)|)$ time. 

\begin{definition}{(traversal sectors)}
Let $x \in \RR^m$, $y \in \RR^\ell$ and $M = (M_m,M_\ell)$ be a pair of traversal matrices with matching dimension.
For $j\in [\ell]$ 
we define $S^{(x,M)}_j := \{x_i ~|~ i \in [m] \text{ and $M$ matches $x_i$ and $y_j$ }\}.$
We call the sequence $(S^{(x,M)}_1,\dots, S^{(x,M)}_\ell)$ the traversal sectors of $x$ and $M$.
\end{definition}
Furthermore,  $(S_1,\dots, S_\ell)$ are called  traversal sectors of $x$, if there exists a pair of traversal matrices $M$ with traversal sectors $(S_1,\dots, S_\ell)$. 
We observe that for a given pair $M=(M_m, M_\ell)$ of
traversal matrices we get
\[\Vert M_m x - M_\ell y\Vert_\infty = \max_{1 \leq i \leq \ell} \max_{z \in  S^{(x,M)}_i} |z - y_i| = \max_{1 \leq i \leq \ell} \max\{|\min(S^{(x,M)}_i) - y_i|, | \max (S^{(x,M)}_i) - y_i|\}.\]

Thus, to determine the Fr\'echet distance between time series $x$ and $y$ it suffices to consider the minimum and maximum value in each traversal sector. This will be used in the following definition.

\begin{definition}[$\ell$-profile]
Let $\ell \in \NN$ and 
and $x\in \RR^m$.
For an arbitrary $y\in \RR^\ell$ and
any pair of traversal matrices $M = (M_m,M_\ell)$
of matching dimension we call the sequence 
\[ 
\big(\rank_x(\min (S^{(x,M)}_1)), \rank_x(\max (S^{(x,M)}_1))\big), \dots, \big(\rank_x(\min (S^{(x,M)}_\ell)), \rank_x(\max (S^{(x,M)}_\ell))\big)
\] the $\ell$-profile of $(x,M)$. 
\end{definition}

\begin{figure}[H]
    \centering
    \includegraphics[width=0.48\linewidth]{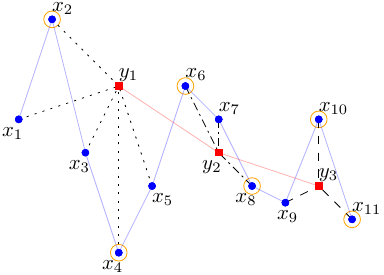}
    \caption{Depicted are two time series $x$ and $y$ and the matched vertices through some pair of traversal matrices $M=(M_{11},M_{3})$. The traversal sectors of $(x,M)$ are $S^{(x,M)}_1 = \{x_1,x_2,x_3,x_4,x_5\}, S^{(x,M)}_2 =\{x_6,x_7,x_8\}$ and $S^{(x,M)}_3 = \{x_9,x_{10},x_{11}\}$. The orange highlighted vertices are the extrema values located in each traversal sector. The ranks of the values  $(x_2,x_4),(x_6,x_8),(x_{10},x_{11})$ define the $3$-profile of $(x,M)$. }
    \label{fig: example profile}
\end{figure}

See Figure \ref{fig: example profile}
for an example of traversal sectors and $\ell$-profile.
Let $
D_{(x,\ell)}$ denote the set of all $\ell$-profiles of a time series $x$
over all pairs of traversal matrices with matching dimension for time series of complexity $m$ and time series of complexity $\ell$. 

Since the information stored in an $\ell$-profile of some time series $x$ and pair of traversal matrices $M = (M_m,M_\ell)$ preserves $\Vert M_m x - M_\ell y\Vert_\infty$ for any $y \in \RR^\ell$ we can argue that two time series $x$ and $x'$ with $\vectorset(x) = \vectorset(x')$
that have the same set of $\ell$-profiles are interchangeable w.r.t. the discrete Fréchet distance. This is stated in the following lemma.

\begin{restatable}{lemma}{lemmatchingprofiles}\label{lem: matching profiles}
Let $m,m',\ell \in \NN$.
Let $x\in \RR^{m}$ and $x'\in \RR^{m'}$ be two time series with $\vectorset(x) = \vectorset(x')$ 
and with the same set of $\ell$-profiles, i.e. 
$D_{(x,\ell)} = D_{(x',\ell)}$.
 Then for every $y \in \RR^\ell$ we have $d_{dF}(x,y) = d_{dF}(x',y).$
\end{restatable}

\begin{proof}       

    Since $D_{(x,\ell)} = D_{(x',\ell)}$, for any pair of traversal matrices $M = (M_m,M_\ell)$ there is a pair of traversal matrices $M'= (M'_{m'},M'_\ell)$  s.t. for all $i \in [\ell]$,  $\max(S^{(x,M)}_i) = \max(S^{(x',M')}_i)$ and $\min(S_i^{(x,M)}) = \min(S^{(x',M')}_i)$ and vice versa.

    Consider some arbitrary $y \in \RR^\ell$ and let $M = (M_{m},M_\ell)$ be a pair of traversal matrices s.t. $\Vert M_{m} x - M_\ell y\Vert_\infty = d_{dF}(x,y)$.
Let $M' = (M'_{m'},M'_\ell)$ be a pair of traversal matrices for $x'$ and $y$ such that $(x',M')$ has the same $\ell$-profile as $(x,M)$.

    Then we get
    \begin{align*}
    d_{dF}(x,y) &= \Vert M_{m} x - M_\ell y \Vert_\infty \\
    &= \max_{i \in [\ell]} \max \{|\min(S^{(x,M)}_i) - y_i|, | \max(S^{(x,M)}_i) - y_i|\}\\
    &= \max_{i \in [\ell]} \max \{|\min(S^{(x',M')}_i) - y_i|, | \max(S^{(x',M')}_i) - y_i|\}\\
    &= \Vert M_{m'}' x' - M_\ell' y \Vert_\infty \\
    &\ge d_{dF}(x',y).
    \end{align*}

 For the other direction let $M'=(M_{m'}',M_\ell')$ be a pair of traversal matrices s.t. $\Vert M_{m'}' x' - M_\ell' y\Vert_\infty = d_{dF}(x',y)$. 
    Let $M = ( M_{m}, M_\ell)$ be a pair of traversal matrices such that $(x,M)$ has the same $\ell$-profile as $(x',M')$. Then we get
    \begin{align*}
        d_{dF}(x',y) &= \Vert M_{m'}' x' - M_\ell' y \Vert_\infty \\
                 &= \max_{i \in [\ell] } \max \{|\min(S^{(x',M')}_i) - q_i|, | \max(S^{(x',M')}_i) - q_i|\}\\
                 &= \max_{i \in [\ell] } \max \{|\min(S^{(x,M)}_i) - y_i|, | \max(S^{(x,M)}_i) - y_i|\}\\
                 &= \Vert M_{m} x - M_\ell y\Vert_\infty\\
                 &\geq d_{dF}(x,y),
    \end{align*}
    which concludes the proof.
\end{proof}

Observe that the profiles of a time series only encode the ranks of the values of the time series. Thus, whenever two time series have the same sequence of ranks they have the same set of profiles. This is stated in the following lemma.

\begin{restatable}{lemma}{lemsameprofileset}
\label{lemma:sameprofileset}
Let $X=\{x_1,\dots, x_{t}\}$ be a set with $x_1< x_2< \dots < x_t$. Let $x=(x_{i_1},\dots, x_{i_m})$ be a time series whose values are taken from $X$ and that satisfies $\vectorset(x) = X$.
Consider a different set $Y=\{y_1,\dots, y_{t}\}$ with 
$y_1< y_2 < \dots < y_t$. 
Let $y=(y_{i_1},\dots, y_{i_m})$ be a time series whose values are taken from $Y$ and that satisfies $set(y) = Y$. Then it holds for any $\ell \in \NN$ that $D_{(x,\ell)} = D_{(y,\ell)}$.
\end{restatable}

\begin{proof}
Consider an arbitrary $\ell$-profile $p \in D_{(x,\ell)}$ then by definition of $p$ there exists a pair of traversal matrices $M$ and $(S^{(x,M)}_1,\dots, S^{(x,M)}_\ell)$ s.t. $p$ equals
\[
\big( (\rank_x(\min(S^{(x,M)}_1)), \rank_x(\max(S^{(x,M)}_1))), \dots, ( \rank_x(\min(S^{(x,M)}_\ell), \rank_x(\max(S^{(x,M)}_\ell))\big).
\]
As $y$ has the same complexity as $x$ we know that 
$(S^{(y,M)}_1,\dots, S^{(y,M)}_\ell)$
are traversal sectors. By the assumptions of $x$ and $y$, $\rank_x(x_{i_j}) = \rank_y(y_{i_j})$ for $j \in [m]$, which implies $\rank_x(\min(S^{(x,M)}_k)) = \rank_y(\min(S^{(y,M)}_k)$ and $\rank_x(\max(S^{(x,M)}_k)) = \rank_y(\max(S^{(y,M)}_k)$, for $k \in [\ell]$. Thus, $p$ is an $\ell$-profile for $(y,M)$ and $p \in D_{(y,\ell)}$.
The other direction is analogous, which concludes the proof.
\end{proof}

We can now show that for every time series $x$ with $|\vectorset(x)| \leq r$ and $\ell \in \NN$ there is a common complexity $f(r,\ell)$ such that there is a time series $y$ of complexity at most $f(r,\ell)$ that has the same set of $\ell$-profiles as $x$.
\begin{restatable}{lemma}{lemexistencematchingprofiles}\label{lem: existence matching profiles} 
There exists a function $f:\NN \times \NN \rightarrow \NN$ such that for every 
$m \in \NN$, every
$x\in \RR^m$ with $|\vectorset(x)| \le r$ and every $\ell \in \NN$, there exists $y \in \RR^{f(r,\ell)}$ such that for every $q\in \RR^\ell$ we have $D_{(y,\ell)} = D_{(x,\ell)}$, $set(x) = set(y)$ and, in particular,
$d_{dF}(x,q) = d_{dF}(y,q)$.  
\end{restatable} 

\begin{proof}
Define $\mathcal D_{r,\ell,i} =\{D :
\exists x\in \RR^i \text{ with $|\vectorset(x)| \le r$ and set of $\ell$-profiles } D\}$ and let $\mathcal D_{r,\ell} = \bigcup_i \mathcal D_{r,\ell,i}$.
Next define function $f:\NN \times \NN \rightarrow \NN$ as 
$f(r, \ell) = \max_{D\in \mathcal D_{r,\ell}} \min\{l \in \NN : \exists x\in \RR^{l} \text{ with set of } \ell\text{-profiles } D\}$ for $\ell\ge 3$ and $f(r,\ell) :=f(r,3)$ for $\ell=1$ and $\ell=2$. Now consider an arbitrary 
$x\in \RR^m$ with $|\vectorset(x)|\le r$ and set of $\ell$-profiles $D$. By definition, we have $D\in \mathcal D_{r,\ell}$ and so we know that there exists a $y' \in \RR^{f(r,\ell)}$ with set of profiles $D$. We also observe that for $\ell\ge 3$ every element of $\vectorset(x)$ appears as the only element in some subsequence $S_i^{(x,M)}$ in some $\ell$-profile. 
This implies that $|\vectorset(y')| = |\vectorset(x)|$. Then for $\ell \ge 3$ we can apply Lemma \ref{lemma:sameprofileset} and replace every occurence of the $i$-th largest element of $\vectorset(y')$ in $y'$ by the $i$-th largest element in $\vectorset(x)$ to 
obtain a vector $y$ with $\vectorset(x) = \vectorset(y)$ and set of profiles $D$. By Lemma \ref{lem: matching profiles}
it follows that $d_{dF}(x,q) = d_{dF}(y,q)$ for all $q \in \RR^\ell$.

To solve $\ell =2$ and $\ell =1$ we observe that in this case any $q \in \RR^\ell$ can be transformed to a $q'\in \RR^3$ by copying the last value. This modification preserves the discrete Fréchet distance and so the lemma follows in this case as well.
\end{proof}

The following corollary results from combining Lemma \ref{lem: reduce value domain} and Lemma \ref{lem: existence matching profiles}.

\begin{corollary} 
For every $\varepsilon \in (0,1)$, $\ell \in \NN$ there is a $d_0=d_0(\varepsilon, \ell)$ such that 
for every $m \in \NN$ and every time series $x\in \RR^m$ there exists a 
time series $y
\in \RR^{d_0}$ such that for every $z\in \RR^\ell$ we have
\begin{equation*}
(1-\varepsilon) d_{dF}(x,z) \le d_{dF}(y,z) \le (1+\varepsilon) d_{dF}(x,z).
\end{equation*}
\end{corollary}

\begin{proof}
By Lemma \ref{lem: reduce value domain} 
we know that there is $r = r(\varepsilon,\ell) \in O(\ell/\varepsilon)$ such that we
can compute a time series $x' \in \RR^m$ with $|\vectorset(x')|\le r$ s.t. for all $z \in \RR^{\ell}$
\begin{equation*}
(1-\varepsilon) d_{dF}(x,z) \le d_{dF}(x',z) \le (1+\varepsilon) d_{dF}(x,z).
\end{equation*} 
Let $f$ be the function as defined in Lemma \ref{lem: existence matching profiles}. By Lemma \ref{lem: existence matching profiles} we get that there
exists $y \in \RR^{f(|\vectorset(x')|,\ell)}$
s.t. for all $z \in \RR^{\ell}$, $d_{dF}(x',z)= d_{dF}(y,z)$. Now define 
$d_0(\varepsilon,\ell) = \max_{1\le i \le r}f(i,\varepsilon)$. 
We observe that if the dimension of $y$ is smaller than $d_0(\varepsilon,\ell)$ then we can repeat the first
element of $y$ in order to get a 
vector of dimension $d_0(\varepsilon,\ell)$ without affecting the discrete Frechet distance. 
In total we get that for all $z \in \RR^\ell$, 
\begin{equation*}
(1-\varepsilon) d_{dF}(x,z) \le d_{dF}(y,z) \le (1+\varepsilon) d_{dF}(x,z).
\end{equation*}
\end{proof}

\subsection{Algorithm}

In this section we turn the previously discussed existential result into an algorithmic one. We have to deal with the fact that we do not know the bound on the complexity $f(|\vectorset(x)|,\varepsilon)$ - we only know that it exists.

Given some time series $x \in \RR^m$ the solution will be to iterate over all  increasing complexities $t$ and all time series
in $\vectorset(x)^t$ until we find one that has the set same of  $\ell$-profiles as $x$.
Our previous results guarantee that this algorithm terminates with 
 $t \leq f(|\vectorset(x)|,\varepsilon)$. By Lemma \ref{lem: existence matching profiles} we have that $d_{dF}(x,q) = d_{dF}(y,q)$, for all $q \in \RR^{\ell}$. 
The challenge is to efficiently compute the set of $\ell$-profiles for a given time series, which we will discuss in the remaining part of the section.

To compute the set $D_{(x,\ell)}$, for some $x \in \RR^m$, we iterate over all sequences $p \in ({|\vectorset(x)|^2})^\ell$ of potential $\ell$-profiles and decide if there exists a pair of traversal matrices $M$ s.t. $p$ is an $\ell$-profile of $(x,M)$. To do so it is sufficient to decide the existence of the corresponding traversal sectors $S^{(x,M)}_{1},\dots, S^{(x,M)}_{\ell}$.

 Consider a sequence $p = ((p^1_1,p^2_1),\dots,(p^1_\ell,p^2_\ell)) \in ({|\vectorset(x)|^2})^\ell$  and let $\min_i$ be the element of $\vectorset(x)$ with rank $p^1_i$ and $\max_i$ be the element of $\vectorset(x)$ with rank $p^2_i$, for $1\leq i \leq \ell$. Note that $\min_i$ and $\max_i$ are supposed to be the minimum and maximum element according to the considered $\ell$-profile. The objective is to decide if there exist traversal sectors
$S_1,\dots,S_\ell$ of $x$ that are consistent with $p$, i.e. for all $i \in [\ell]$, $\min_i,\max_i \in S_i$ and $S_i \subseteq [\min_i,\max_i]$. This is done in a recursive way and stated as dynamic program in the form of Algorithm \ref{alg: Partition of Profile}.

The idea is as follows. For some $t \in [m]$ and $h \in [\ell]$ we would like to know if there are  traversal sectors $S_1,\dots, S_h$ for $(x_1,\dots,x_t)$ that are compatible with $p$.
However, in order to set up a recursion, we also need to know whether the minimum and/or maximum value of $S_h$
has already appeared in $x_1,\dots, x_t$.
For this purpose, we introduce two Boolean variables $a$ and $b$.

Concretely, Algorithm~\ref{alg: Partition of Profile} is a dynamic programm that checks for the existence of partial solutions of the following form. 

\begin{definition}[compatible traversal sectors]
 Given a time series $x=(x_1,\dots,x_m) \in \RR^m$, sequence  $p=\big((\min_1,\max_1),\dots,(\min_\ell,\max_\ell)\big) \in (\vectorset(x)^2)^\ell$, $h \in [\ell]$, $t \in [m]$.
 Traversal sectors $S_1,\dots,S_h$ of $(x_1,\dots,x_t)$ are compatible with $p$ if they satisfy: 
  \begin{enumerate}
    \item $\min_i,\max_i \in S_i$ for $1\leq i\leq h-1$
    \item $S_i \subseteq [\min_i,\max_i]$ for $1\leq i \leq h$
    \end{enumerate}
\end{definition}

Note that the last sector $S_h$ is not required to fulfill the conditions of the profile, in that we don't require $\min_h$ and $\max_h$ to be contained in $S_h$. This is needed as they may not have appeared in the sequence at that time of the algorithm, yet.
Lemma~\ref{lem:partitionDP_invariant} below shows correctness of Algorithm~\ref{alg: Partition of Profile}. The crucial observation is that partial solutions can be combined as follows. 

\begin{observation}\label{obs:combination_ab}
For any $h \in [\ell], t \in [m]$, let $S_1,\dots,S_h$ of $(x_1,\dots,x_t)$ be compatible traversal sectors of  $(x_1,\dots,x_t)$ with $\min_h \in S_h$ and let $S'_1,\dots,S'_h$ be compatible traversal sectors of  $(x_1,\dots,x_t)$ with $\max_h \in S'_h$, then either $S_h \subseteq S'_h$ and therefore also $\min_h \in S'_h$, or $S'_h \subseteq S_h$ and therefore also $\max_h \in S_h$. This follows because the last sector in each case consist of the elements of a prefix of the same sequence.
\end{observation}

\begin{algorithm}
\caption{\label{alg: Partition of Profile}}
\begin{algorithmic}[1]
\Procedure {\texttt{AssignmentDP}}{$(x_1,\dots,x_m)$, $((\min_1,\max_1),\dots,(\min_\ell,\max_\ell))$}
\State $\feas \gets \{\text{False}\}^{(\ell+1)\times (m+1)}$, $\hmin \gets \{\text{False}\}^{(\ell+1)\times (m+1)}$, $\hmax \gets \{\text{False}\}^{(\ell+1)\times (m+1)}$
\State $\feas[0,0] \gets \text{True}$, $\hmin[0,0] \gets \text{True}$, $\hmax[0,0] \gets \text{True}$
\For{$h = 1 \textbf{ to } \ell $}
    \For{$t = 1 \textbf{ to } m $}
       \If{$x_t \in [\min_h,\max_h]$}
             \If{$\feas[h,t-1]$}
                  \State $\feas[h,t] \gets \text{True}$
                  \If{$\hmin[h,t-1] \lor x_t = \min_h$} 
                     \State $\hmin[h,t] \gets \text{True}$
                  \EndIf
                   \If{$\hmax[h,t-1] \lor x_t = \max_h$} 
                     \State $\hmax[h,t] \gets \text{True}$
                  \EndIf
             \Else \If{$\feas[h-1,t-1] \land \hmin[h-1,t-1] \land \hmax[h-1,t-1]$}
                  \State $\feas[h,t] \gets \text{True}$
                  \If{$x_t = \min_h$} 
                     \State $\hmin[h,t] \gets \text{True}$
                  \EndIf
                   \If{$x_t = \max_h$} 
                     \State $\hmax[h,t] \gets \text{True}$
                  \EndIf                    
                  \EndIf
             \EndIf           
       \EndIf
    \EndFor
\EndFor
\State\Return $\feas[\ell,m] \land \hmin[\ell,m] \land \hmax[\ell,m]$

\EndProcedure
\end{algorithmic}
\end{algorithm}

\begin{restatable}{lemma}{lemassignmentdp}\label{lem:partitionDP_invariant}
Given $x \in \RR^m$ and $p \in {(\vectorset(x)^2)}^\ell$.
Let $D,a$ and $b$ be the tables constructed during Algorithm \ref{alg: Partition of Profile} with input $(x,p)$.
For every $h \in [\ell], t \in [m]$, after the corresponding iteration of the for-loop, we have $D[h,t] = True$
if and only if there exist 
compatible traversal sectors $S_1,\dots,S_h$ of $(x_1,\dots,x_t)$. In addition, we have that $a[h,t] = True$ (resp. $b[h,t] = True$) if and only if there exists such a compatible solution with $\min_h \in S_h$ (resp. $\max_h \in S_h$).
\end{restatable}

\begin{proof}
We prove the lemma by induction on the iterations of the inner for-loop. For the base case consider the first iteration with $h=1$ and $t=1$. Assume $x_1 \in [\min_1,\max_1]$. In this case, the clause in line 6 evaluates to True, but the clause in line 7 evaluates to False. However, the clause in line 14 evaluates to True, because these Booleans were set in line 3 to True to initialize the algorithm. In this case, $D[1,1]$ is set to True, which is correct since $S_1=\{x_1\}$ is a compatible traversal sector. Furthermore, $a[1,1]$ and $b[1,1]$ are set correctly. Otherwise, if $x_1 \notin [\min_1,\max_1]$, then there exists no compatible traversal sectors and $D[1,1],a[1,1],$ as well as $b[1,1]$ remain set to False.

For the induction step consider any $h,t \geq 1$ with $h >1$ or $t > 1$. Assume $x_t \in [\min_h,\max_h]$. 
If $D[h,t-1]$ is True, then by induction there exist compatible traversal sectors $S_1,\dots,S_h$ for $(x_1,\dots,x_{t-1})$. Therefore there exist compatible traversal sectors  of $(x_1,\dots,x_t)$ by adding $x_t$ to $S_h$. Furthermore, $a[h,t]$ and $b[h,t]$ are set correctly.

Otherwise, if $D[h,t-1]$ is False, then we check in line 14, if
 $a[h-1,t-1]$ and $b[h-1,t-1]$ are True. By induction and Observation \ref{obs:combination_ab} this is the case if and only if there exist compatible traversal sectors $S_1,\dots,S_{h-1}$ for $(x_1,\dots,x_{t-1})$ and it holds that $\min_{h-1} \in S_{h-1}$ as well as $\max_{h-1} \in S_{h-1}$.
 As such, there exist compatible traversal sectors $S_1,\dots,S_{h}$ with $S_h=\{x_t\}$ for $(x_1,\dots,x_{t})$. Furthermore, $a[h,t]$ and $b[h,t]$ are set correctly with respect to $S_h$. 

Now, assume $x_t \notin [\min_h,\max_h]$. In this case, there exists no compatible traversal sectors and $D[h,t],a[h,t],$ as well as $b[h,t]$ remain set to False.
\end{proof}

The next lemma relates the output of Algorithm \ref{alg: Partition of Profile} to the existence of an $\ell$-profile.

\begin{lemma} \label{lem: check profile feasible}
Given $x \in \RR^m$ and $p =((p^1_1,p^2_1),\dots, (p^1_\ell,p^2_\ell)) \in {|(\vectorset(x)^2)|}^\ell$. Let
\[p' = \big( (min_1, max_1),\dots, (min_\ell, max_\ell)\big),\]
where $min_i,max_i \in \vectorset(x)$ s.t. $rank_x(min_i) = p^1_i$ and $rank_x(max_i) = p^2_i$, for $in \in [\ell]$. Then Algorithm \ref{alg: Partition of Profile} with input $(x,p')$ takes  $O(m\ell)$ time and returns True iff  there exists a pair of traversal matrices $M$ s.t. $p$ is an  $\ell$-profile for $(x,M)$.
\end{lemma}

\begin{proof}
Let $D, a$ and $b$ be the tables constructed by Algorithm \ref{alg: Partition of Profile} and assume 
the algorithm returns True. Then, it must be that $a[\ell,m]$ and $b[\ell,m]$, as well as $D[\ell,m]$ are all set to True.
By Lemma \ref{lem:partitionDP_invariant} and Observation~\ref{obs:combination_ab},  there exist traversal sectors $S_1,\dots,S_\ell$ of $x$ that are compatible with $p$ and it holds that $\min_\ell \in S_\ell$ and $\min_\ell \in S_\ell$. By definition there exists a pair of traversal matrices $M$ s.t. $S^{(x,M)}_i = S_i$, for $1\leq i \leq \ell$ implying that $p$ is an $\ell$-profile of $(x,M)$.

Next assume that there exists a pair of traversal matrices $M$ s.t. $p$ is an $\ell$-profile for $(x,M)$. Then  $S^{(x,M)}_1,\dots,S^{(x,M)}_\ell$ are traversal sectors for $x$ that are compatible with $p$ and it holds that $\min_\ell \in S^{(x,M)}$ and $\max_\ell \in S^{(x,M)}$. By Lemma \ref{lem:partitionDP_invariant}, 
$a[\ell,m]$ and $b[\ell,m]$, as well as $D[\ell,m]$ are correctly set to True.

The initialization of  $D$ takes $O(\ell m)$ time and the algorithm takes $O(\ell m)$ many iterations of the inner for-loop, which require constant time each.
\end{proof}

\begin{algorithm}
\caption{\label{alg: ComplexityReduction}}
\begin{algorithmic}[1]
\Procedure {\texttt{ComplexityReduction}}{$x \in \RR^m$, $\ell \in \NN$, $\varepsilon \in \RR_{\geq 0}$}

\State $x'\gets \text{ReduceValueDomain}(x,\ell,\varepsilon)$

\State Let $t$ be smallest value s.t. it exists $z \in \vectorset(x')^t$ with $D_{(z,\ell)} = D_{(x',\ell)}$  

\State \Return $z$

\EndProcedure
\end{algorithmic}
\end{algorithm}

\begin{theorem}
\label{theorem:dimreduction}
Let $\varepsilon \in (0,1)$ and $\ell \in \NN$. There exists a $d_0 = d_0(\varepsilon, \ell)$ such that for arbitrary $m\in \NN$ and input time series $x \in \RR^m$ Algorithm \ref{alg: ComplexityReduction} computes in time $O(m \ell \log^2 m) + O(\ell/\eps)^{2 d_{0}}$  
a time series $z \in \RR^{d_0}$ s.t. for all $y \in \RR^\ell$,
\begin{equation*}
    (1-\varepsilon) d_{dF}(x,y) \leq d_{dF}(z,y) \leq (1 + \varepsilon)  d_{dF}(x,y).
\end{equation*}
\end{theorem}
\begin{proof}
    Let $x'$ be the time series returned by ReduceValueDomain$(x,\ell,\varepsilon)$. Then by Lemma \ref{lem: reduce value domain} for all $y \in \RR^\ell$ 
    \[(1-\varepsilon) d_{dF}(x,y) \leq d_{dF}(x',y) \leq (1+\varepsilon) d_{dF}(x,y)\]
    with  $|\vectorset(x')| \leq r_0$, where $r_0 \in O(\ell/\varepsilon)$.  
    By Lemma \ref{lem: existence matching profiles} there exists a function $f:\NN \times \NN \rightarrow \NN$ s.t. there exists 
    a time series $z \in \vectorset(x')^{f(|\vectorset(x')|,\ell)}$, with $D_{(z,\ell)} = D_{(x',\ell)}$ and $\vectorset(x') = \vectorset(z)$,
    which further implies by Lemma \ref{lem: matching profiles} that for arbitrary $y\in \RR^\ell$ we have $d_{dF}(z,y) = d_{dF}(x',y)$. It follows that  \[(1-\varepsilon) d_{dF}(x,y) \leq d_{dF}(z,y) \leq (1+\varepsilon) d_{dF}(x,y).\]

    By \Cref{lem: reduce value domain}, the time needed to compute $x'$ is in $O(\ell m\log^2 m)$.   
    Then, by \Cref{lem: check profile feasible}, computing $D_{(x',\ell)}$ takes $O(\ell/{\eps})^{2\ell}\cdot m$ time by enumerating all $|\vectorset(x')|^{2\ell}$ candidate $\ell$-profiles and checking if they are valid $\ell$-profiles for $x'$ using \Cref{alg: Partition of Profile}. 
    To find the time series $z$ with smallest complexity with its vertices in $\vectorset(x')$ such that $D_{(z,\ell)} = D_{(x',\ell)}$, we enumerate all vectors over 
    $\vectorset(x')$ in increasing length until we find a vector with the same set of $\ell$-profiles as $x'$. For each increasing value of $i=1,2,\ldots,t$, for each vector $z\in \vectorset(x')^i$ we compute its set of $\ell$-profiles, by enumerating all $|\vectorset(x')|^{2\ell}$ candidate $\ell$-profiles and checking if they are valid $\ell$-profiles for $z$ using \Cref{alg: Partition of Profile}.  By definition, we have $t\leq d_0 := \max(\max_{1\le r \le r_0} f(r,\ell), 2\ell )$. 
    The overall running time of this step is in $O(d_0^2)\cdot O(\ell/\eps)^{d_0}\cdot O(\ell/\eps)^{2\ell} \subseteq O(\ell/\eps)^{2d_0}$. 
\end{proof}

    We remark that although the algorithm is constructive, the bound on the running time in the above theorem is only existential. We do not have an explicit bound for the dependence on $\varepsilon$ and $\ell$.
    
\section{Near Linear Time Approximation Scheme for $(k,\ell)$-Median}
In this section, we take all necessary steps to make the linear time approximation scheme by Cohen-Addad et al. \cite{CFS21} compatible with our problem. Since their algorithm requires a pre-specified set from which the centers are chosen, we need to efficiently compute a bounded set of candidate centers that approximately covers the set of realizable solutions. Then, it remains to argue that the metric space after the dimension reduction has a doubling dimension that is independent of the original dimension. 

Throughout this section, we adopt the representation of traversals as ordered sequences of pairs of indices, as defined in \Cref{section:preliminaries}, as it provides the most natural framework for our purposes. In particular, in certain technical arguments where the enumeration of traversals is required, the alternative representation of traversal matrices proves less convenient (and less efficient in terms of space usage).  

\subsection{Computing Candidate Centers}
We begin with an auxiliary lemma on discretizing the set of time series that have complexity $\ell$ and are within distance $r$ from an arbitrary time series of complexity $m$. This was essentially proven and used by Filtser et al.~\cite{FFK23} in the context of approximate nearest neighbor data structures. Since there is no standalone lemma with the exact same statement in~\cite{FFK23}, we include a proof for completeness. 

\begin{lemma}
\label{lemma:candidate_single_scale}
Given a time series $x\in \RR^m$, its minimum-error $\ell$-simplification $\tilde{x}$, $r>0$ and $\eps \in (0,1)$, we can compute a set of time series $S_{x,r,\eps}$, in time $O(1/\eps)^{\ell}$, such that $|S_{x,r,\eps}| \in O(1/\eps)^{\ell}$ and for all $y\in \RR^{\ell}$, if $d_{dF}(x,y)\leq r$ then there exists $y'\in S_{x,r,\eps}$ such that $d_{dF}(y',y)\leq \eps r$. 
\end{lemma}

\begin{proof}
    Let $y = (y_1,\ldots, y_{\ell})$, $\tilde{x}=(\tilde{x}_1,\ldots, \tilde{x}_{\ell})$ and $\mathcal{G}_{\eps r}$ be the regular grid with cell width $r\varepsilon$. We define  ${I}_j := [\tilde{x}_j- (2r+\eps), \tilde{x}_j+ (2r+\eps)]$ and $\tilde{I}_j := I_j\cap \mathcal{G}_{\eps r}$, for any $j \in [\ell]$.
    For each traversal $T$ of two time series with complexity $\ell$
    we compute a set of time series $S_T$ as follows: for each $j\in [\ell]$ let $i_{j}$ be the index of the first vertex matched to the vertex at index $j$ in the traversal, i.e., $i_{j}= \min\{i:~ (i,j)\in T \}$
    . We compute $S_T = \prod_{j=1}^{\ell}\tilde{I}_{i_j}$ and we output $S_{x,r,\eps} = \bigcup_{T} S_T$. 
    
    For a fixed traversal $T$, there are $O(1/\eps)^{\ell}$ combinations of vertices defining the time series included in $S_T$. By \cite{FFK23}[Lemma 4], the number of different traversals is at most $4^{\ell}$. Hence, the running time of computing  the set $S_{x,r,\varepsilon}$ and its size  are upper bounded by $O(1/\eps)^{\ell}$.
    
    To show correctness, we first observe that if $d_{dF}(x,y)\leq r$, then by the triangle inequality $d_{dF}(\tilde{x},y) \leq d_{dF}(x,y)+d_{dF}(x,\tilde{x})\leq 2r$. Hence, there is an optimal traversal $T^{\ast}$ matching the vertices of $y$ with vertices of $\tilde{x}$ with cost at most $2r$, which implies that there exists a time series in $S_{T^{\ast}}$ with vertices $y_1',\ldots,y_{\ell}'$ such that for any $j\in [\ell]$, $|y_j'-y_j| \leq \eps r$. Therefore, $d_{dF}(y,y')\leq \eps r$. 

\end{proof}
We now proceed by computing a set of candidate centers, i.e., a set of time series with complexity $\ell$, from which the $k$ centers can be chosen, while only sacrificing an arbitrarily small approximation factor from the cost of the optimal solution. 
\begin{lemma}
\label{lemma:mediancandidates}
    Given a set of time series $P \subset \RR^m$ and $k,\ell \in \NN$, and $\eps \in (0,1)$, we can compute in time $O(m^3 n\ell)  + O(n\log n)\cdot O(1/\eps)^{\ell}$ a set $S \subset \RR^\ell$ of size $O( n \log n) \cdot O(1/\eps)^{\ell}$ such that there exists a set $C \subset S$, $|C|=k$, that satisfies 
    \[
    \sum_{x\in P} \min_{c\in C}d_{dF}(x,c) \leq (1+3\eps)\Delta^{\ast},
    \]
    where $\Delta^{\ast}$ is the optimal $(k,\ell)$-median cost.  
\end{lemma}

\begin{proof}
    We first compute $\Delta$ such that $\Delta^* \leq \Delta \leq 12 \cdot \Delta^*$ using the algorithm from \cite{buchin2019hardness}[Theorem 10]. Then, for each $x \in P$ we compute a minimum-error $\ell$-simplification $\tilde{x}$ using \cite{bereg2008simplifying}. For each $ i \in \{0,\dots , 3+\lceil \log(n)\rceil\}$, let $r_i = \frac{2^i\cdot \Delta}{12 n}$ and for each $x\in P$, compute $S_{x,r_i,\eps}$ using \Cref{lemma:candidate_single_scale} with $\tilde{x}$. We output $S=\bigcup_{x\in P}\bigcup_{i=1}^{3+\lceil \log(n)\rceil} S_{x,r_i,\eps}$.

    The time needed to compute $\Delta$ is in  $O(m^3 + mn\ell(m + \log n \log(n/\ell)))$~\cite{buchin2019hardness}. 
    The time needed to compute all simplifications using \cite{bereg2008simplifying} is $O(n\ell  m  \log m  \log(m/\ell))$. Given the simplifications, by \Cref{lemma:candidate_single_scale}, the running time to create $S$ and its size is in $O(n\log n) \cdot (1/\eps)^{\ell}$ since we consider $O(\log n)$ different parameters $r_i$. 

    To show correctness, consider any time series $c\in C^{\ast}$, where $C^{\ast}$ is an optimal solution for $(k,\ell)$-median. For any $x\in P$ that has $c$ as its closest center from $C^*$ it holds that $d_{dF}(x,c)\leq \Delta^{\ast}$. 
    Now let $x_c$ be the time series in $P$ which has smallest discrete Fr\'echet distance to $c$  (ties are broken arbitrarily)  and let $i^{\ast}$ be the smallest value $i$ such that 
    $d_{dF}(x_c,c)\leq r_i$. By \Cref{lemma:candidate_single_scale}, $S_{x_c,r_{i^\ast},\eps}$ contains a time series $c'\in \RR^{\ell}$ such that $d_{dF}(c,c')\leq \eps r_{i^{\ast}}$. 
    If $i^{\ast}=0$, then $d_{dF}(c,c') \leq \frac{\eps \Delta^{\ast}}{n}$ and by the triangle inequality, for any $x\in P$, $d_{dF}(x,c')\leq d_{dF}(x,c)+ \frac{\eps \Delta^{\ast}}{n}$. 
    If $i^{\ast}>0$, then by the triangle inequality, for any $x\in P$ that has $c$ as its closest center from $C^*$,
    \begin{align*}
        d_{dF}(x,c')&\leq d_{dF}(x,c)+d_{dF}(c,c')\\
        &\leq  d_{dF}(x,c)+ \eps \cdot r_{i^{\ast}}\\
        &\leq  d_{dF}(x,c)+ 2\eps \cdot d_{dF}(x_c,c)\\
        & \leq (1+2\eps)\cdot  d_{dF}(x,c).
    \end{align*}
 Hence, for any $c\in C^{\ast}$ there is a $c' \in S$ such that for any $x\in P$ that has $c$ as its closest center, $d_{dF}(x,c') \leq \max ( (1+2\eps)\cdot  d_{dF}(x,c) ,  d_{dF}(x,c)+ \frac{\eps \Delta^{\ast}}{n})$. 
 Now, let $C$ be a set containing one such $c'$ for each $c \in C^{\ast}$. The cost of this solution is
 \begin{align*}
     \sum_{x\in P} \min_{c'\in C}d_{dF}(x,c') 
     &\leq 
     \sum_{x\in P} \min_{c\in C^{\ast}} \max \left (  (1+2\eps)\cdot  d_{dF}(x,c) ,  d_{dF}(x,c)+ \frac{\eps \Delta^{\ast}}{n} \right )
     \\
     &\leq  \sum_{x\in P} \min_{c\in C^{\ast}}   \left((1+2\eps)\cdot  d_{dF}(x,c) + \frac{\eps \Delta^{\ast}}{n}\right)
     \\
     &\leq  (1+2\eps)\cdot  \sum_{x\in P} \min_{c\in C^{\ast}}    d_{dF}(x,c) + n\cdot \frac{\eps \Delta^{\ast}}{n}
     \\
     &\leq (1+3\eps)\Delta^{\ast}.    
 \end{align*}

\end{proof}

\subsection{Doubling Dimension Reduction}
In this section, we present standard results concerning the doubling dimension of time series, under the discrete Fr\'echet distance, which, when combined with our dimension reduction technique, facilitate the efficient application of the algorithm of Cohen-Addad et al. \cite{CFS21}. Essentially, our dimension reduction technique can be used to effectively reduce the doubling dimension of the input time series. 
\begin{definition}{(Doubling Dimension)}
Consider any metric space $(X,d_{X})$ and for any $x\in X$ let $B_X(x,r)=\{y\in X \mid d_{X}(x,y)\leq r\}$. The {\em doubling constant} of $X$, denoted $\lambda_X$, is the smallest integer $\lambda_X$ such that for any $x \in X$ and $r > 0$, the ball $B_X(x, r)$ can be covered by at most $\lambda_X$ balls of radius $r/2$ centered at points in $X$. The {\em doubling dimension} of $(X,d_{X})$ is defined to be equal to $\log \lambda_X$.  
\end{definition}

The following statement about the doubling dimension of the discrete Fr\'echet distance is folklore, but we include a proof for completeness. 
\begin{prp}
\label{prop:ddim}
The metric space defined on the equivalence classes of time series with complexity at most $m$ that have pairwise  discrete Fr\'echet distance $0$,  
equipped with the discrete Fr\'echet distance, has doubling dimension $\Theta(m)$.
\end{prp}
\begin{proof}
Note that any time series with complexity less than $m$ can be expanded to a time series of complexity $m$ by repeating the first vertex of the element, without affecting the discrete Fr\'echet distance.
We first show that the doubling dimension is at least $m$. Consider a time series $x = (x_1,\dots,x_m)$  such that $x_i = 3\cdot i $ for each $i\in [m]$. Let $B$ be the discrete Fr\'echet ball of radius $1$ centered at $x$, i.e. $B = \{y \in \RR^m~|~ d_{dF}(x,y) \leq 1\}$. All time series defined by sequences in $\prod_{i=1}^m \{x_i-1,x_i+1\}$ are in $B$, but no two of them can be covered by the same ball of radius $1/2$. Hence, we need at least $2^m$ balls of radius $1/2$ to cover $B$, implying that the doubling dimension is at least $m$.

For the upper bound, consider any time series $x=(x_1,\ldots,x_m)$ and any radius $r>0$. Let $B_r$ be the discrete Fr\'echet ball of radius $r$ centered at $x$. For each $i\in [m]$,
let $a_i := x_i - r/2$, $b_i := x_i + r/2$. Let  $T$ be the optimal traversal between $x$ and an arbitrary time series $y\in B_r$ of complexity $m$.  For each $j\in [m]$, let $i_{j}$ be the first index (pointing to the $i_{j}$th vertex of $x$) paired with the $j$th point of $y$ in $T$. It must hold 
$y_j \in [x_{i_j} - r,x_{i_j}]$
or 
$y_j \in [x_{i_j}, x_{i_j}+r]$.
Hence, the set $\prod_{j=1}^m \{a_j,b_j\}$ contains a time series within discrete Fr\'echet distance $r/2$ from $y$. By taking into account all traversals and for each traversal all relevant combinations of points $a_i,b_i$, we cover the entire $B_r$ with balls of radius $r/2$. By \cite{FFK23}[Lemma 4], 
the number of traversals is at most $4^m$, and for each traversal, we consider at most $2^m$ time series as centers of the balls of radius $r/2$, implying an upper bound of $O(m)$ on the doubling dimension. 
\end{proof}

\subsection{Putting Everything Together}

The framework of Cohen-Addad et al. \cite{CFS21} provides a near-linear time approximation algorithm for the $k$-median problem in metrics of low doubling dimension. 
Formally they define the problem as follows. For a set of clients $C$, a set of candidate centers $F$, an integer $k$, and
a function $\chi :~ C \mapsto \{1, \ldots , n\}$ the goal is to minimize $\sum_{c\in C} \chi(c) \cdot \min_{f\in S} dist(c, f)$. Note that the function $\chi$ can be utilized to encode multiplicities of the clients.

The following statement is slightly rephrased from \cite{CFS21}[Theorem 1]. 

\begin{theorem}[\cite{CFS21}]
\label{thm:cohenaddadetal}
    For any $0 < \eps < 1/3$, set of clients $C$ and set of candidate centers $F$, there exists a randomized $(1+\eps)$-approximation algorithm for $k$-median in metrics of doubling dimension d with running time $2^{{(1/\eps)}^{O(d^2)}}
n \log^4 n+ 2^{O(d)}n \log^9 n$ and success probability at least $1 - 2\eps$, where $n = |C \cup F|$ and the time needed to access any distance is assumed to be constant.
\end{theorem}

Our solution builds upon \Cref{thm:cohenaddadetal} as follows. 
First, we reduce the complexity $m$ of the input time series to a complexity that only depends on $\ell,\eps$ using \Cref{alg: ComplexityReduction}, to obtain a set of clients $C$ with constant complexity. Then, 
we compute a set of candidate centers $F$ using \Cref{lemma:mediancandidates}, which has a size near linear in $n$, and independent of $m$. 
Then \Cref{prop:ddim} implies that the doubling dimension of the resulting space is upper bounded by a function of $\ell,\eps$. Strictly speaking, the Fr\'echet distance is a pseudo-metric as distance between distinct curves can be $0$. Therefore, one has to consider a metric space defined on the equivalence classes of curves with pairwise distance $0$. Each equivalence class has a corresponding representative time series and, for some given time series $x$ its representative can be computed by removing all duplicates of neighboring values in $x$. 

Now the algorithm of \Cref{thm:cohenaddadetal} can be run in time that is linear in $n$,$m$ and importantly without any exponential dependence on $k$. The pseudocode can be found in \Cref{{alg:ltas}}.

\begin{algorithm}
\caption{\label{alg:ltas}}
\begin{algorithmic}[1]

\Procedure {\texttt{NLTAS}}{$P\subset \RR^m$, $k,\ell \in \NN$, $\eps \in (0,1)$}
\State $C \gets $ for each $x\in P$, run \Cref{alg: ComplexityReduction} with input $x,\ell,\eps/16$ and store output in $C$
\State $F\gets$ output of the algorithm of \Cref{lemma:mediancandidates} with input $C,k,\ell$ and $\eps/12$.
\State $Sol \gets$ output of the algorithm of \Cref{thm:cohenaddadetal} with input $
C,F,k,\eps/4 $
\State \Return $Sol$
\EndProcedure
\end{algorithmic}
\end{algorithm}

\begin{theorem}
\label{theorem:final_formal}
Let $\eps\in (0,1/2]$ and $\ell \in \mathbb N$ be constants. Given a set $P$ of $n$ real-valued time series of complexity $m$, and parameter $k$ Algorithm \ref{alg:ltas} computes in time
\[
O( nm \log^2 m +  n \log^{10} n )
\]
and with success probability at least $1-\eps$ a $(1+\eps)$-approximation to the $(k,\ell)$-median problem under discrete Fréchet distance. 
\end{theorem}
\begin{proof}
By \Cref{theorem:dimreduction}, \Cref{alg: ComplexityReduction} runs in $O(m \ell \log^2 m) + O(\ell/\eps)^{2 d_{0}}$ time and outputs 
a time series $z \in \RR^{d_0}$, where $d_0 = d_0(\varepsilon, \ell)\geq 2\ell$. Hence, computing $C\subset \RR^{d_0}$ costs 
$O(nm \ell \log^2 m) + n\cdot O(\ell/\eps)^{2 d_{0}}$ 
time. 
The algorithm of \Cref{lemma:mediancandidates} runs in time $O(d_0^3 n\ell)  + O(n\log n)\cdot O(1/\eps)^{\ell}$ and outputs a set $F \subset \RR^\ell$ of size $O( n \log n) \cdot O(1/\eps)^{\ell}$. 
Finally, we run the algorithm of \Cref{thm:cohenaddadetal} on $C,F$, where $|C\cup F| \in  O(n\log n)\cdot O(1/\eps)^{\ell}$. Both $C$ and $F$ can be seen as subsets of $\RR^{d_0}$, since the time series of complexity $\ell$ can be expanded while preserving all distances, by duplicating the last vertex sufficiently many times. Hence, the doubling dimension of the resulting ambient space is in $O(d_0)$ by \Cref{prop:ddim}. By \Cref{thm:cohenaddadetal} and assuming constant time distance evaluations, the running time of this step is 
\[
O\left(\frac{1}{\eps}\right)^{\ell}\cdot 
2^{{(1/\eps)}^{O\left(d_0^2\right)}} O(n \log^5 n)+ O\left(\frac{1}{\eps}\right)^{\ell}\cdot O( 2^{O(d_0)}n \log^{10} n ).
\]
Since $\eps, \ell$ and $d_0$ are constants distances can be computed in constant time and furthermore we can simplify the running time to \[
O( nm \log^2 m +  n \log^{10} n ).
\]

To show correctness, first observe that by \Cref{theorem:dimreduction}, the cost of any solution for $P$ is preserved up to a factor of $(1\pm\frac{\eps}{16})$ after reducing the complexity of each $x\in P$ resulting in $C$. The optimal solution for $C$ corresponds to a solution for $P$ which is at most $\frac{1+\eps/16}{1-\eps/16}\leq 1+\eps/4$ from the optimal. 
Now by \Cref{lemma:mediancandidates}, there will be a solution consisting of medians from $F$ which has cost at most $(1+\frac{\eps}{4})$ times that of the optimal solution for $C$. Therefore, the execution of the algorithm of \Cref{thm:cohenaddadetal} on $C,F$ will find a solution consisting of $k$ time series from $F$ that has a cost of at most  $(1+\frac{\eps}{4})^2$ times that of the optimal solution for $C$. This solution has a cost of at most  $(1+\frac{\eps}{4})^3 \leq 1+\eps$ times that of the optimal solution for $P$. 
\end{proof}

\section{Coreset for $(k,\ell)$-median}

In this section, we discuss an additional implication of our dimension reduction result in the context of constructing coresets for $(k,\ell)$-median. Given a set $P$ of $n$ time series, an $\eps$-coreset of $P$ for the $(k,\ell)$-median problem, is a weighted set $S\subseteq P$ such that for any $C\subset \RR^{\ell}$, $|C|=k$, 
\[(1-\eps)\sum_{x \in P}\min_{c \in C} d_{dF}(x,c)\leq \sum_{x \in S} w_x \cdot \min_{c \in C}  d_{dF}(x,c) \leq (1+\eps)\sum_{x \in P}\min_{c \in C} d_{dF}(x,c) ,\] 
where $w_x$ is the weight associated with $x$. 

A recent result of Cohen-Addad et al.~\cite{CDRSS25}~[Corollary 7.2] implies coresets for the $(k,\ell)$-median problem under the discrete Fr\'echet distance of size $\tilde{O}(\eps^{-2} k \ell \log(m) )$.  By combining this result with \Cref{theorem:dimreduction}, we obtain coresets of size $\tilde{O}(\eps^{-2}k \ell \log (d_0))$, i.e., completely independent of the size of the input since $d_0$ is a function of $\ell,\eps$. The result is formally stated as follows.
\begin{corollary}
        Let $\varepsilon \in (0,1)$ and $k,\ell \in \NN$. There exists  $d_0 = d_0(\varepsilon, \ell)$ such that for any set $P$ of time series there exists an $\eps$-coreset for the $(k,\ell)$-median problem of size $\tilde{O}(\eps^{-2}k \ell \log (d_0))$. 
\end{corollary}

\section{Conclusions}

We have presented the first near-linear time approximation scheme for $(k,\ell)$-median clustering under discrete Frechet distance when the ambient dimension of the input time series is $1$. 
One obvious open question is to extend the result to time series of higher ambient dimension. The main challenge here is to extend the dimension reduction. Another interesting open problem is to give a constructive upper bound on the target dimension. 
\bibliography{biblio}
\bibliographystyle{plainurl}

\end{document}